\documentclass{article}

\usepackage{mathtools}
\usepackage{lmodern}
\usepackage[short, nocomma]{optidef}
\usepackage{subcaption}
\usepackage{amsthm}
\usepackage{xcolor}

\usepackage{setspace}
\usepackage[margin=2cm]{geometry}
\usepackage{newtxtext}
\usepackage{newtxmath}

\usepackage[absolute]{textpos}

\makeatletter
\newcommand*{\defeq}{\mathrel{\rlap{%
			\raisebox{0.3ex}{$\m@th\cdot$}}%
		\raisebox{-0.3ex}{$\m@th\cdot$}}%
	=}

\makeatletter
\def\thmheadbrackets#1#2#3{%
	\thmname{#1}\thmnumber{\@ifnotempty{#1}{ }\@upn{#2}}%
	\thmnote{ {\the\thm@notefont[#3]}}}
\makeatother

\newtheoremstyle{brakets}
{}
{}
{\normalfont}
{}
{\bfseries}
{.}
{ }
{\thmheadbrackets{#1}{#2}{#3}}

\newtheoremstyle{defbrakets}
{}
{}
{\normalfont}
{}
{\bfseries}
{.}
{ }
{\thmheadbrackets{#1}{#2}{#3}}

\newtheoremstyle{defproblem}
{}
{}
{\normalfont}
{}
{\bfseries}
{.}
{ }
{\thmheadbrackets{#1}{#2}{#3}}

\theoremstyle{brakets}
\newtheorem{thm}{Theorem}

\theoremstyle{defbrakets}

\newtheorem{lem}[thm]{Lemma}
\newtheorem{defn}[thm]{Definition}
\newtheorem{prop}{Proposition}
\newtheorem{exmp}{Example}
\theoremstyle{defproblem}

\newtheorem{rem}[thm]{Remark}



\newcommand{\z}[1]{\textcolor{black}{#1}}

\begin{document}

	\title{\bfseries Traffic Divergence Theory: An Analysis Formalism for Dynamic Networks}
\author{Matin~Macktoobian\footnote{matin.macktoobian@mail.utoronto.ca}, Zhan Shu, and Qing Zhao}%
\date{Electrical and Computer Engineering Department\\ University of Alberta\\Edmonton, AB, Canada}
\maketitle

\begin{textblock}{14}(6,1)
	\noindent\textbf{\color{red}Published in 
		``IEEE Access'' \\DOI: 10.1109/ACCESS.2024.3383436
		}
\end{textblock}

\begin{abstract}
Traffic dynamics is universally crucial in analyzing and designing almost any network. This article introduces a novel
theoretical approach to analyzing network traffic dynamics. This theory's machinery is based on the notion of traffic divergence, which captures the flow (im)balance of network nodes and links. It features various analytical probes to investigate both spatial and temporal traffic dynamics. In particular, the maximal traffic distribution
in a network can be characterized by spatial traffic divergence rate, which reveals the relative difference among node traffic
divergence. To illustrate the usefulness, we apply the theory to two network-driven problems: throughput estimation of data center networks and power-optimized communication planning for robot networks, and show the merits of the proposed theory through simulations.
\end{abstract}

\textbf{keywords}: Dynamic Networks, Network Traffic Dynamics, Traffic Divergence

\maketitle
\doublespacing
\section{Introduction}
\label{sec:introduction}
Network traffic is critical in various network operations such as routing, congestion control, traffic anomaly detection, planning and scheduling, etc. The ubiquitous importance of these operations in almost any class of networks, from computer and data center networks to software and ad-hoc ones, dictates the necessity of powerful tools to represent network traffic, their flows, and distribution. Given the existing modeling strategies in the literature, one may mainly classify those traffic models into two categories based on their generality and applicability to different types of networks and operations.
\subsection{Related Work}
\z{The first traffic model category is often reformulated and specifically customized from one application domain to another. The major bulk of these strategies belongs to self-similarity seeking methods \cite{barakat2003modeling,jelenkovic1995algorithmic}. Following the idea of the similar reaction to similar patterns of traffic, these models mostly rely on various probabilistic distributions \cite{lee2005stochastic} to model traffic as fractional Brownian motion and even geometrical intuitions, e.g., those borrowed from the theory of fractals  \cite{nogueira2003modeling}. Then, one usually applies observer operators to those distributions to determine metrics of interest by measuring desired states of a network. The efficient scalability of these models comes at the cost of their relatively inaccurate assessments. The larger a network becomes, the more problematic such inaccuracies will become, particularly in the case of performance factors like latency. Moreover, self-similarity conventions constructing these models assume some degrees of relative homogeneity associated with traffic distributed in a network. This assumption may generally not be valid in the majority of networks involving structural and/or data heterogeneity, including computer \cite{dudin2003optimal}, social \cite{xia2018exploring}, ad-hoc \cite{umer2018dual}, and software \cite{nishanbayev2019model} ones. Machine learning techniques have also been explored to identify various traffic-driven aspects of networks \cite{macktoobian2021data,macktoobian2023topology}. For instance, Principal Component Analysis \cite{abdi2010principal} was applied to structures of complex networks to infer traffic-related conclusions via reducing their dimensionality \cite{lakhina2004structural}. Recurrent neural networks applied to dynamic states of computer networks, in the form of time series, are employed for congestion control forecasting and anomaly detection \cite{madan2018predicting}. These methods are overall flexible in design, but transferring traffic knowledge between different areas of complex networks while preserving highly-accurate predictions is challenging \cite{mohammed2011survey, papadogiannaki2021survey}.}

\z{In contrast to the specificity of self-similarity methods, network flow theory \cite{latham2006network,ford1956network} is an attempt to model network traffic generally. This formalism, derived from the graph theory, was originally invented to solve the maximum flow problem \cite{han2014maximum}, i.e., finding the largest net flow from a source node to a destination one. Different graph-based versions of this problem were solved by Dinic's algorithm \cite{tarjan1986algorithms}, EK algorithm \cite{lammich2016formalizing}, MPM algorithm \cite{malhotra1978v}, and Orlin algorithm \cite{orlin2013max}. Due to the limited tools in the network flow theory (residuals, augmenting paths, etc.), the cited algorithms act on graphs globally. Hence, they cannot be effectively scaled for exceedingly-large graphs such as those of data center topologies and social networks. Furthermore, this theory does not exhibit any machinery to analyze temporal and/or spatial dynamics of network traffic distribution. A complex version of the maximum flow problem, i.e., the multi-commodity flow problem \cite{kumar2018semi} to which various system-level constraints may be applied, has emerged as a fundamental hurdle in computing different network characteristics such as throughput. Namely, the NP-complete nature of this problem makes it difficult to be used to compute traffic-based metrics of large-scale networks.}
\subsection{Contribution}
In this article, we propose a novel concept, i.e., \textit{Traffic Divergence (TD)}, and related formalism that not only features a general solid tool to model traffic flow and distribution in networks but also reveals new implications about spatial and/or temporal traffic dynamics in such networks. Compared to the reviewed strategies above, here are the major advantages and merits exhibiting the contributions of this article.
\begin{itemize}
	\item This theory is inherently suitable for analysing traffic- and flow-based variations in networks. Accordingly, it provides a unified approach to modeling network traffic dynamics for incoming and outsourcing flows of both nodes and links.
	\item Contrary to the network flow theory, our theory yields global traffic-driven conclusions about large-scale networks based on their local dynamisms in node neighborhoods. This localizability feature particularly contributes to realizing computationally efficient and scalable traffic dynamical analyses.
	\item Our traffic divergence theory is generally applicable to a wide range of networks. Particularly, this theory is expressive enough to model different problem statements and constraints associated with various networks such as computer and ad-hoc networks. In this regard, we later supply examples of how our theory may be applied to traffic-related problems associated with data center networks and ad-hoc robot networks.
\end{itemize}
\subsection{Organization}
This article is organized as follows. Section \ref{sec:mach} and \ref{sec:max} introduce the network model and the formalism of TD. Section \ref{sec:app} is devoted to the discussion of two network applications that our theory contributes to the improvement of their state-of-the-art solutions. Namely, in Section \ref{subsec:throu}, we employ our theory to develop a congestion-aware routing algorithm for data center topologies. Section \ref{subsec:adhoc} utilizes the Traffic Divergence Theory to formulate an optimization problem to address power-optimized communication planning for ad-hoc robot networks. \z{Section \ref{sec:disc} discuss some limitations of our work based on which new potential streams of research are proposed.} Our conclusions are provided in Section \ref{sec:conc}.
\section{The Machinery of Traffic Divergence}
\label{sec:mach}
In this section, we establish a flow-based formalism to provide an analytical tool for various network analyses at the heart of which traffic distribution is a pivotal concept. For this purpose, we first establish the general network model on which TD is developed.
\subsection{Network Model}
\label{subsec:netmod}
\z{In our theory, a \emph{network} is an undirected graph $\left(\mathscr{V},\mathscr{E}\right)$ consisting of a set of nodes $\mathscr{V}$ and a set of links $\mathscr{E}\subseteq\mathscr{V} \times \mathscr{V}$ connecting them. For a node $u \in \mathscr{V}$, its neighbor set is defined as 
	\begin{equation}
		\mathscr{N}_{u} \defeq \{z\in\mathscr{V}\!\setminus\!\{u\} \mid \mathscr{A}_{uz} \neq 0 \}, 
	\end{equation}	
	where $\mathscr{A}$ denotes the adjacency matrix of the graph (\cite{weisstein2007adjacency}).  For a node $u \in \mathscr{V}$ and a node $v \in \mathscr{N}_{u}$, $f_{uv}\left(t,\Delta\right)$ ($t,\Delta \in \mathbb{R}^+$) denotes the positive \emph{average flow rate} from $u$ to $v$ over the time period $\left[t,t+\Delta\right]$, e.g., the amount of materials transported or the data packets transmitted divided by $\Delta$, as shown in Figure \ref{fig:flow}. For the case that $\Delta \to 0^+$, $f_{uv}\left(t,\Delta\right)$ should be regarded as the positive \emph{instantaneous flow rate} at time $t$. In the remaining part of the paper, $f_{uv}\left(t,\Delta\right)$ is abbreviate to $f_{uv}$ for notational simplicity unless the time period $\Delta$ is of interest. A \emph{route} from node $u \in \mathscr{V}$ to node $v \in \mathscr{V}$ in a network is a path connecting $u$ and $v$ in which all nodes and links are distinct. It is assumed that the length of any route, the number of links involved, in the network is finite. Also, self-loops or cycles, i.e., links connecting a node to itself, are excluded in the network.}

\z{Each node is supposed to have limited processing capacity, resulting in unequal or imbalanced input and output. For example, $n_{1}$ packets are transmitted into a switch in a data network over a period of time, but only $n_2$ packets, where $n_{2} < n_{1}$, flow out from the switch over the same period. The remaining $n_1-n_2$ packets are either stored in the switch for future processing or discarded due to congestion. We will utilize this imbalance to characterize network traffic as shown later.}

\z{Each link is assumed to have limited transmission capacity. More specifically, when $\Delta$ is fixed, $f_{uv}\left(t,\Delta\right)$ is uniformly bounded for all $t\in \mathbb{R}^+$ and any feasible $u$ and $v$.}
\subsection{Node, Link, and Route TD}
Inspired by the notion of sink (input) and source (output) flows in the network theory (see, e.g., \cite{megiddo1974optimal,ahlswede2000network,carmi2008transport}), we develop the concept of TD for nodes, links, and routes in this subsection.
\begin{figure}
	\centering\includegraphics[scale=1.8]{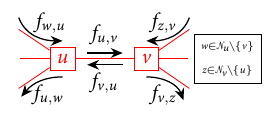}
	\caption{Traffic flow notation associated with nodes of a network.}
	\label{fig:flow}
\end{figure}
\begin{defn}[Node TD]
	\label{defn:d}
	Given a network with node set $\mathscr{V}$, the TD of node $u \in \mathscr{V}$, denoted by $\nabla_{u}$, reads as the difference of the sink flows entering the node and the source flows leaving it, i.e., 
	\begin{equation}
		\nabla_{u} \defeq \sum_{z \in \mathscr{N}_{u}}^{} [f_{zu} - f_{uz}].
	\end{equation}
\end{defn}
\begin{exmp}
	\label{exmp:1}
	Suppose a node set $\mathscr{V}=\{u, z_{1}, z_{2}\}$ and flows $f_{uz_{1}} = 4$, $f_{u,z_{2}} = \sin t$, $f_{z_{1},u} = \cos t$, and $f_{z_{2},u} = 1$. Then, according to Definition \ref{defn:d}, the TD of $u$ can be calculated as $\nabla_{u} = (1+\cos t) - (\sin t+4) = \cos t -\sin t -3$.
\end{exmp}
\begin{defn}[Link TD]
	\label{defn:ltd}
	Given a network with node set $\mathscr{V}$, suppose that $u \in \mathscr{V}$ and $v \in \mathscr{N}_{u}$. Then, the TD of the link connecting $u$ and $v$, denoted as $\nabla_{u,v}$, is defined as the difference of the net flows entering the link and the net flows departing it, namely, 
	\begin{equation}
		\nabla_{u,v} \defeq \smashoperator{\sum_{\substack{w\in\{u,v\},\\ z\in \mathscr{V}\setminus\{u,v\}}}} [f_{zw} - f_{wz}].
	\end{equation}
\end{defn}
\begin{exmp}
	\label{exmp:2}
	Assume that the node set of a network is $\mathscr{V}=\{u, v, z_{1}, z_{2}, z_{3}\}$ associated with flows $f_{u,z_{1}} = 1$, $f_{u,z_{2}} = t^2/e^t$, $f_{z_{1},u} = 2$, and $f_{z_{3},v} = \sin t$, according to Definition \ref{defn:ltd}, the TD of the link connecting $u$ to $v$ can be calculated as $\nabla_{u,v} = (2+\sin t) - (1+t^2/e^t) = 1+\sin t-t^2/e^t$.
\end{exmp}
The proposition below shows that the TD of a link can be factorized as the sum of the TD of corresponding nodes.
\begin{prop}[Node-Link TD Correspondence]
	\label{prop:ltc}
	Given a network with node set $\mathscr{V}$, suppose that $u \in \mathscr{V}$ and $v \in \mathscr{N}_{u}$. Then, the TD of the link connecting $u$ and $v$ can be factorized as 
	\begin{equation}
		\nabla_{u,v} = \nabla_{u} + \nabla_{v}.
	\end{equation}
\end{prop}
\begin{proof}
	See, Appendix \ref{app:ltc}.
\end{proof}
In view of the two definitions above, the TD of a route may be defined as the difference of the external net flows entering its nodes and the net flows departing those nodes.
\begin{defn}[Route TD]
	Given a network with node set $\mathscr{V}$, suppose $\Omega$ is a route including a set of nodes $\mathscr{R} \subseteq \mathscr{V}$. Then, the traffic divergence of $\Omega$, written as $\nabla_{\Omega}$, is defined as 
	\begin{equation}
		\nabla_{\Omega} \defeq \sum_{u\in\mathscr{R}, z\in \mathscr{V}\setminus \mathscr{R}} [f_{zu} - f_{uz}].
	\end{equation}
\end{defn}
\begin{exmp}
	\label{exmp:3}
	Let $\mathscr{V}= \{u, v, w, z_{1}, z_{2}, z_{3}\}$ be the node set of a network. Assume a set of flows as $f_{u,z_{1}} = 1$, $f_{z_{2},v} = 2$, and $f_{w,z_{3}} = 3$. Then, the TD of the route $\Omega$ with nodes ${u,v,w}$  $\nabla_{\Omega}= 2 - (1+1) = 0$.
\end{exmp}
Similar to Proposition \ref{prop:ltc}, the result below demonstrates how the TD of a route can be expressed in terms of the TD of its constituting nodes.
\z{\begin{prop}[Node-Route TD Correspondence]
		\label{prop:rrtd}
		For a route $\Omega$, its TD can be factorized as 
		\begin{equation}
			\nabla_{\Omega} = \sum_{u\in\Omega}\nabla_{u}.
		\end{equation}
\end{prop}}
\begin{proof}
	See, Appendix \ref{app:rrtd}.
\end{proof}
Following this proposition, the TD of a route can be recursively computed throughout its expansion via adding further nodes given the TD of some initial nodes as revealed below.
\z{\begin{prop}[Node-Link-Route TD Correspondence]
		\label{prop:nlrtd}
		Suppose that $\Omega$ is a route between node $u$ and node $v$ and the associated node set is $\{u,w,v\}$. Then, we have 
		\begin{equation}	
			\nabla_{\Omega} = \nabla_{u,w} + \nabla_{v}=\nabla_{w,v} + \nabla_{u}.
		\end{equation}
\end{prop}}
\begin{proof}
	See, Appendix \ref{app:nlrtd}.
\end{proof}
\subsection{Spatial TD Dynamics}
A natural question in network traffic analysis is how the traffic at one specific node affects that of other nodes. Suppose flows $f_{u,z_{1}} = \cos t$, $f_{u,z_{2}} = \cos t$, $f_{z_{1},u} = \sin t$, and $f_{z_{2},u} = \sin t$. Then, it is easy to obtain that $\nabla_{u}=-2\nabla_{z_1}=-2\nabla_{z_2}$, which shows the spatial relationship among different node TD. Loosely speaking, if $\nabla_{z_1}$ or $\nabla_{z_2}$ changes one unit, then $\nabla_{u}$ will change two. Inspired by this, we introduce a new operation to characterize the relative spatial change of TD below.{\begin{defn}[Spatial Node-Node TD Derivative]
		\label{defn:STDD}
		Given a network with node set $\mathscr{V}$, suppose that $u,v \in \mathscr{V}$, and $
		\nabla_{u}=h\left(\nabla_{v}\right)$, where $h$ is a many-to-one mapping. Then, the spatial TD derivative of node $u$ with node $v$ is defined as $
		\frac{\partial \nabla_{u}}{{\partial \nabla_{v}}} \defeq h^{\prime}$,	where $h^{\prime}$ represents the derivative of $h$ with respect to its argument. Furthermore, $\frac{\partial^{n} \nabla_{u}}{{\partial \nabla_{u}}^{n}} \defeq h^{\left(n\right)}$, where $n \in \mathbb{Z}^{+}$ and $h^{\left(n\right)}$ represents the $n$th derivative of $h$ with respect to its argument.
	\end{defn}
	As link TD can be represented by the involved node TD, Spatial Link-Node TD Derivative, $\frac{\partial \nabla_{u,v}}{{\partial \nabla_{u}}}$ or $\frac{\partial \nabla_{u,v}}{{\partial \nabla_{v}}}$, can be defined in a similar way, thus omitted here for brevity. If $h$ in the above definition is not bijective, then $\frac{\partial \nabla_{v}}{{\partial \nabla_{u}}}$ may not be well-defined. For this singular case, we use $\left(\frac{\partial \nabla_{u}}{{\partial \nabla_{v}}}\right)^{-1}$ to calculate $\frac{\partial \nabla_{v}}{{\partial \nabla_{u}}}$. The theorem below asserts how the spatial node-node TD derivatives of two adjacent nodes relates to the related spatial link-node TD derivative.
	\z{	\begin{thm}[Spatial TD Dynamics]
			\label{thm:dtd}
			Let $u\in \mathscr{V}$ and $v \in \mathscr{N}_{u}$. Then, given any $n \in \mathbb{Z}^{+}$, we have
			\begin{equation}
				\label{eq:xxxx}
				\frac{\partial^{n}\nabla_{u,v}}{{\partial \nabla_{u}}^{n}} - \frac{\partial^{n}\nabla_{u,v}}{{\partial \nabla_{v}}^{n}} = \frac{\partial^{n} \nabla_{v}}{{\partial \nabla_{u}}^{n}} - \frac{\partial^{n} \nabla_{u}}{{\partial \nabla_{v}}^{n}}.
			\end{equation}
	\end{thm}}
	\begin{proof}
		See, Appendix \ref{app:dtd}.
	\end{proof}
	\begin{exmp}
		Assume a link including two adjacent nodes $u$ and $v$ such that the traffic divergence of one of them is a spatio-temporal function of that of another, say,
		$
		\nabla_{u} = [m{\nabla_{v}}+n]g(t),
		$ where $m, n \in \mathbb{R}$, and $g(t)$ is a real-valued function in time. From this it follows that $\nabla_{v}=(\frac{\nabla_{u}/g(t) - n}{m})$.
		Then, $\nabla_{u,v}$ ca be expressed as a function of $\nabla_{u}$ or $\nabla_{v}$ as $
		\nabla_{u,v} = \nabla_{u} + (\frac{\nabla_{u}/g(t) - n}{m})$ or $
		\nabla_{u,v} = \nabla_{v} + [m{\nabla_{v}}+n]g(t)
		$. Simple calculations show that equation \ref{eq:xxxx} holds.
	\end{exmp}
	We end this subsection by introducing a concept which will be used later for traffic distribution analysis.
	\begin{defn}[Spatial TD Rate]
		Given a node $u \in \mathscr{V}$, the spatial TD rate of $u$ is defined as 
		\begin{equation}
			\square_{u} \defeq \sum_{z \in \mathscr{N}_{u}}\frac{\partial\nabla_{u}}{\partial\nabla_{z}}.
		\end{equation}
	\end{defn}
	Intuitively, the spatial TD rate of a node characterizes the coupling between the node and its neighbors. If $\square_{u}=0$, then overall the traffic in the neighbor nodes has little impact on that of $u$
	\subsection{Temporal TD Dynamics}
	Traffic divergences usually vary in the course of time, as well, because of variable flows of communication between different pairs of nodes. For the characterization of such a temporal traffic dynamism, we define the notion of temporal traffic divergence rate as below.
	\begin{defn}[Temporal Traffic Divergence Rate]
		Given a node $u \in \mathscr{V}$, the temporal traffic divergence of $u$ is defined as 
		\begin{equation}
			\boxplus_{u} \defeq \sum_{z \in \mathscr{N}_{u}}[\frac{\partial}{\partial t}\nabla_{z}].
		\end{equation}
	\end{defn}
	\begin{exmp}
		\label{examp:pp}
		Assume a node $u$ whose corresponding neighbor set is $\mathscr{N}_{u}$. Let the traffic divergence of neighbor $z\in \mathscr{N}_u$ be an accumulative function of the traffic divergences of its neighbors, say, $
		\nabla_{z} \defeq \sum_{y \in \mathscr{N}_{z}}[m{\nabla_{y}}+n]g(t)
		$, where $m$ and $n$  are real numbers. So, we have $
		\boxplus_{u} = \frac{\mathrm{d} g(t)}{\mathrm{d} t}\sum_{z \in \mathscr{N}_{u},y \in \mathscr{N}_{z}}[m{\nabla_{y}}+n]
		$.
	\end{exmp}
	The result below shows how the temporal traffic rate of a node is bounded by multiplicative functions of its spatial and temporal traffic divergence rates and the cardinality of its neighbor node set.
	\z{	\begin{thm}[Node Temporal Traffic Rate]
			\label{thm:nttr}
			Given a node $u \in \mathscr{V}$ the cardinality of whose neighbor set $\mathscr{N}_{u}$ is $n\defeq \lvert\mathscr{N}_{u}\rvert$, the temporal traffic rate of that node is bounded as 
			\begin{equation}
				-n^{-1}[\square_{u}\boxplus_{u}]\le \frac{\partial \nabla_{u}}{\partial t} \le n^{-1}[\square_{u}\boxplus_{u}].
			\end{equation}
	\end{thm}}
	\begin{proof}
		See, Appendix \ref{app:nttr}.
	\end{proof}
	The computation of node temporal traffic rate is extremely desired in many traffic-driven analyses, such as anomaly detection \cite{gu2005detecting}, phase transition investigations \cite{cheng2020phase}, self-similarity monitoring in cyber-attacked networks \cite{kotenko2021technique}, etc. However, the exact computation of this quantity in long operations of many networks is often practically infeasible. Alternatively, temporal and spatial traffic divergence rates of a node can be readily computed. Thus, one can employ the result of the theorem above to find a bounded profile for the node temporal traffic rate of any desired node in a network.
	\section{Maximal Traffic Distribution}
	\label{sec:max}
	A natural question in network analysis is to quantify the traffic distribution over the network. In terms of the spatial TD rate defined previously, we propose a new quantity to characterize the relative difference of traffic among different nodes.
	\begin{defn}[Maximal Traffic Distribution]
		The traffic associated with a network including nodes $\mathscr{V}$ is called maximally-distributed if the following condition is fulfilled.
		\begin{equation}
			\label{eq:mtdc}
			\Delta_{u,v}\defeq \frac{\square_{u}}{\square_{v}} = 1 \quad (\forall u,v \in \mathscr{V}).
		\end{equation}
	\end{defn}
	An alternative definition would consider the ratio of absolute values of traffic divergences to define the maximum traffic distribution condition. However, the ratio of rates of traffic divergences inherently better incorporate dynamics of traffic divergences compared to the ratio of absolute values. In this regard, the ratio of rates is directly correlated with the desired equilibrium state of traffic. On the other hand, any static metric, such as the ratio of absolute values, essentially requires some excessive conditions on the rate of traffic divergences, as well. The following proposition gives an equivalent characterization of maximal traffic distribution, which may be more computationally efficient.
	\z{\begin{prop}[A Equivalent Condition]
			\label{prop:mtdc}
			The maximal traffic distribution condition in (\ref{eq:mtdc}) is equivalent to
			\begin{equation*}
				\frac{\displaystyle\sum\limits_{z \in \mathscr{N}_u}\Bigg[\frac{\partial\nabla_{z}}{\partial\nabla_{u}} + \Big[\frac{\partial}{\partial\nabla_{u}} - \frac{\partial}{\partial\nabla_{z}}\Big]\nabla_{u,z}\Bigg]}{\displaystyle\sum\limits_{z \in \mathscr{N}_v}\Bigg[\frac{\partial\nabla_{z}}{\partial\nabla_{v}} + \Big[\frac{\partial}{\partial\nabla_{v}} - \frac{\partial}{\partial\nabla_{z}}\Big]\nabla_{v,z}\Bigg]} = 1\quad (\forall u,v \in \mathscr{V})
			\end{equation*}
	\end{prop}}
	\begin{proof}
		See, Appendix \ref{app:mtdc}.
	\end{proof}
	The assessment of (\ref{eq:mtdc}) is computationally resource-intensive particularly when one deals with large networks. This condition seeks a direct but naive approach to globally investigating the maximality of traffic distribution in a network. However, it is beneficial to demonstrate that the overall localized assessments of the condition in each node's neighborhood condition are indeed equivalent to the cited global assessment, i.e., 
	\begin{equation}
		\begin{split}
			&[\Delta_{u,v} = 1 (\forall u \in \mathscr{V})(\forall v \in \mathscr{N}_{u})] \equiv\\ &[\Delta_{u,v} = 1 (\forall u,v \in \mathscr{V})]
		\end{split}
	\end{equation}.
	\z{\begin{lem}
			\label{lem:symm}
			Suppose nodes $u,v,w \in \mathscr{V}$ such that $v$ is between $u$ and $w$. Then, given a $\gamma \in \mathbb{R}^{+}$, assume that the relative traffic rate in both sides of $v$ are the same, say, $\Delta_{u,v} = \Delta_{w,v} \defeq \gamma$. Then, the overall end-to-end relative traffic rate between $u$ and $w$ is maximally distributed.
	\end{lem}}
	\begin{proof}
		See, Appendix \ref{app:symm}.
	\end{proof}
	The result above implies that the traffic rate symmetrically distributed around the neighboring nodes of a particular node represents the maximal traffic distribution condition in the route connecting those nodes together serially.
	\z{\begin{thm}[Localization]
			\label{thm:loc}
			Given a network comprising some nodes belonging to $\mathscr{V}$, the maximal traffic distribution condition (\ref{eq:mtdc}) may be relaxed to the following equivalent localized version. 
			\begin{equation}
				\Delta_{u,v} = 1 \quad(\forall u \in \mathscr{V})(\forall v \in \mathscr{N}_{u})
			\end{equation}
	\end{thm}}
	\begin{proof}
		See, Appendix \ref{app:loc}.
	\end{proof}
	\begin{rem}
		The advantage of the localized version above is that it entails less number of evaluations corresponding to $\Delta_{u,v} = 1$ compared to those of (\ref{eq:mtdc}).
	\end{rem}
	\begin{rem}
		In practice due to the complexity of traffic dynamics in computer networks, it is generally unlikely that the maximal traffic distribution condition can be met for every neighborhood. Alternatively, given a sufficiently small radius $\epsilon$, one may seek the realization of
		\begin{equation}
			\label{eq:remm}
			\lvert\Delta_{u,v} - 1\rvert \le \epsilon \quad (\forall u\in \mathscr{V})(\forall v\in \mathscr{N}_{u}).
		\end{equation}
	\end{rem}
	
	\section{Example Applications and Numerical Simulations}
	\label{sec:app}
	In this section, we illustrate how our traffic divergence theory can be helpful to model and solve different flow-based network problems. For this purpose, Section \ref{subsec:throu} employs the theory to model the dynamics of throughput in two popular network topologies used in datacenters. This formulation particularly maximizes the distribution of traffic in such topologies to minimize the risk of node/link congestion. The second application, described in Section \ref{subsec:adhoc}, expresses a traffic-driven formulation of a solution to power-optimized communication planning problem in ad-hoc robot networks using Traffic Divergence Theory. For all results presented in this section, we use NetworkX \cite{hagberg2008exploring} for network graph synthesis, and Gekko optimization suite \cite{beal2018gekko} for nonlinear programming.
	\begin{figure}
		\centering\includegraphics[scale=1]{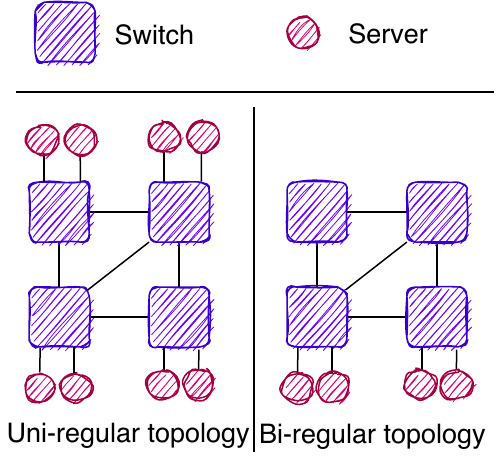}
		\caption{Uni-regular and bi-regular topologies in datacenters. (In the uni-regular case, all switches are connected to (usually similar) number of servers. However, some servers in the bi-regular case are exclusively connected to other switches for routing purposes.)}
		\label{fig:uni-bi}
	\end{figure}
	\subsection{Congestion-Avoiding Throughput Analysis for Datacenter Topologies}
	\label{subsec:throu}
	Throughput maximization is extremely desired in any network topology deployed in datacenters as a well-known measure of communication performance. In these networks, nodes are either switches or servers. Throughput depends on routing in not only computer networks \cite{leonardi2007optimal,johnson2009comparison,ma1998routing} but also wireless networks \cite{sun2015quantifying,dong2010secure,qin2005impact,zeng2008end} and even network-on-chip architectures \cite{ahmed2012xyz,kim2012clumsy,bakhoda2010throughput} and in datacenters \cite{xu2013greening,kabbani2014flowbender}. This dependence stems from the fact that different routing schemes may lead to various latency profiles that negatively correlate with throughput.
	
	Throughput maximization with respect to uni-regular and bi-regular topologies, as depicted in Figure \ref{fig:uni-bi}, has been extensively investigated \cite{jyothi2016measuring,singla2014high,munene2021throughput}. In particular, Jain method \cite{jain2014maximizing}, bisection bandwidth (BBW) method \cite{kloti2015policy}, and TUB method \cite{namyar2021throughput} have been presented as the most successful models yielding close throughput upper bounds with respect to real throughputs of various datacenter topologies. In this section, we propose a congestion-avoiding routing algorithm, based on Traffic Divergence Theory, using which the cited throughput gap reduces compared to the quoted methods. In other words, given a uni-regular or bi-regular topology associated with a datacenter network, we find the throughput profile of the topology according to the maximally-distributed traffic in that network. In this regard, we highlight the congestion awareness of our analysis that directly benefits from the traffic distribution tools of the theory introduced in Section \ref{sec:mach}.
	\begin{figure}
		\centering\includegraphics[scale=.9]{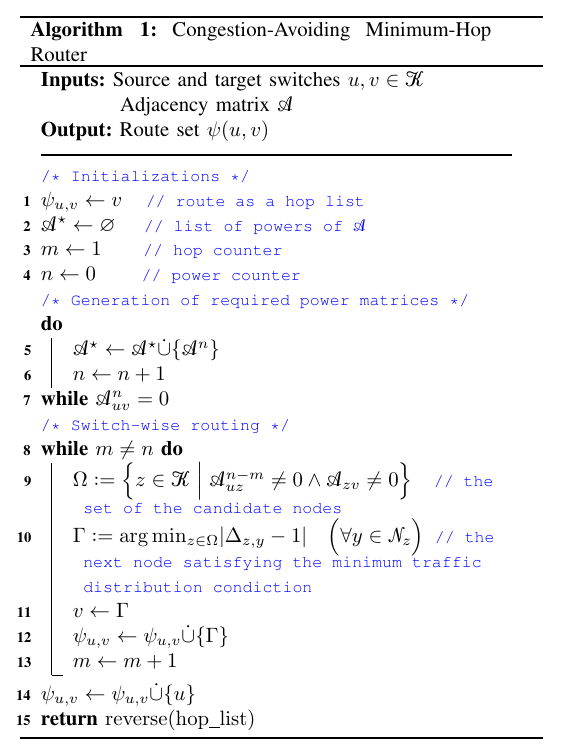}
	\end{figure}
	
	Congestion is a concern attributed to routes and their corresponding nodes. Thus, we intend to conduct our throughput analysis based on a congestion-aware routing strategy. For this purpose, we propose a congestion-avoiding minimum-hop routing algorithm, i.e.,  Algorithm 1, that finds optimal routes between two typical nodes such that their minimum distance and maximum traffic distribution are simultaneously taken into account as a trade-off using our traffic divergence theory. In particular, each route (see, Line 1 of Algorithm 1) is defined as a series of consecutive hops that reach a target node $v \in \mathscr{K}$ from a source node $u \in \mathscr{K}$. The length minimization criterion may be addressed using various powers of the adjacency matrix associated with a network (Lines 5-7). Namely, the $(u,v)$ entry of $\mathscr{A}^{n}$ denotes the number of routes from $u$ to $v$ including exactly $n$ hops. In Lines 9 and 10, the nodes enjoying the least traffic distributed in their neighborhood are found to shape the desired optimal route.
	
	The key idea for throughput analysis associated with a topology is the partitioning of the flow reaching a node to those belonging to the transient traffic and those planned to reach that node. The total profile of transient traffic can also be directly computed with respect to the routes of the topology. Accordingly, we derive a closed-form equation representing the throughputs corresponding to both uni-regular and bi-regular topologies subject to minimum node and/or link traffic divergence.
	\z{\begin{prop}
			\label{prop:cc}
			Given a network including node (switch) set $\mathscr{K}$, let $n\defeq |\mathscr{K}|$ be the dimension of the adjacency matrix associated with that network.Then, the worst-case computational complexity of Algorithm 1 is $\mathcal{O}\big(n^{3}\big)$.
	\end{prop}}
	\begin{proof}
		See, Appendix \ref{app:cc}.
	\end{proof}
	\z{\begin{thm}[Throughput Subject to Maximum Traffic Distribution]
			\label{thm:throughput}
			Given a traffic matrix\footnote{Our doubly-stochastic traffic matrix formulation particularly follows the saturated hode model \cite{duffield1999flexible}. In this model, nodes are considered to be switches rather than servers. However, with modification of traffic matrix definition, one can take a server-level analysis, as well.} $\mathscr{T}$ associated with a uni-regular or a bi-regular topology, let $\mathscr{E}$ be the number of switch-to-switch links in that topology. Denote by $\mathscr{K}$ the set of all switches in the topology. Let and $\mathscr{H}$ and $\mathscr{H}_u$ be the universal number of servers connected to each server in the uni-regular case and the number of servers connected to switch $u \in \mathscr{K}$ in the bi-regular case, respectively. Then, the maximum achievable throughput for the topologies based on the maximal traffic distribution requirement is as follows.
			\begin{equation}
				\label{eq:throughput}
				\theta^\mathscr{T} = \frac{2\mathscr{E}}{\mathscr{H}}\Bigg[1 + \smashoperator[l]{\sum\limits_{\substack{u \in \mathscr{K}, v \in \mathscr{K}\setminus\{u\},\\\rho \in \psi_{u,v}}}} \Bigg[\alpha_{\rho}|\rho| \Big[\sum_{w \in \rho}^{} \nabla_{w}\Big]^{-1}\Bigg] - |\psi_{u,v}| \Bigg]^{-1}
			\end{equation}
	\end{thm}}
	\begin{proof}
		See, Appendix \ref{app:throughput}.
	\end{proof}
	\begin{rem}
		In the theorem above, there is no general constraint to select the elements of the co-efficient set 
		\begin{equation}
			\label{eq:coeff}
			\{1>\alpha_{\rho} \in \mathbb{R}^{\ge 0} \mid\forall \rho \in \psi_{u,v}   \},
		\end{equation}
		because practically not all route selections are based on the shortest ones, and there are numerous potential functional requirements, according to a datacenter and its network size, that may impact that route selection. Thus, in practice, one expects to observe a gap between the real throughput (as the solution to the path-based multi-commodity flow problem \cite{kumar2018semi}) of a network and what (\ref{eq:throughput}) predicts.
	\end{rem}
	\begin{figure*}
		\centering
		\begin{subfigure}{0.5\textwidth}
			\centering\includegraphics[scale=0.6]{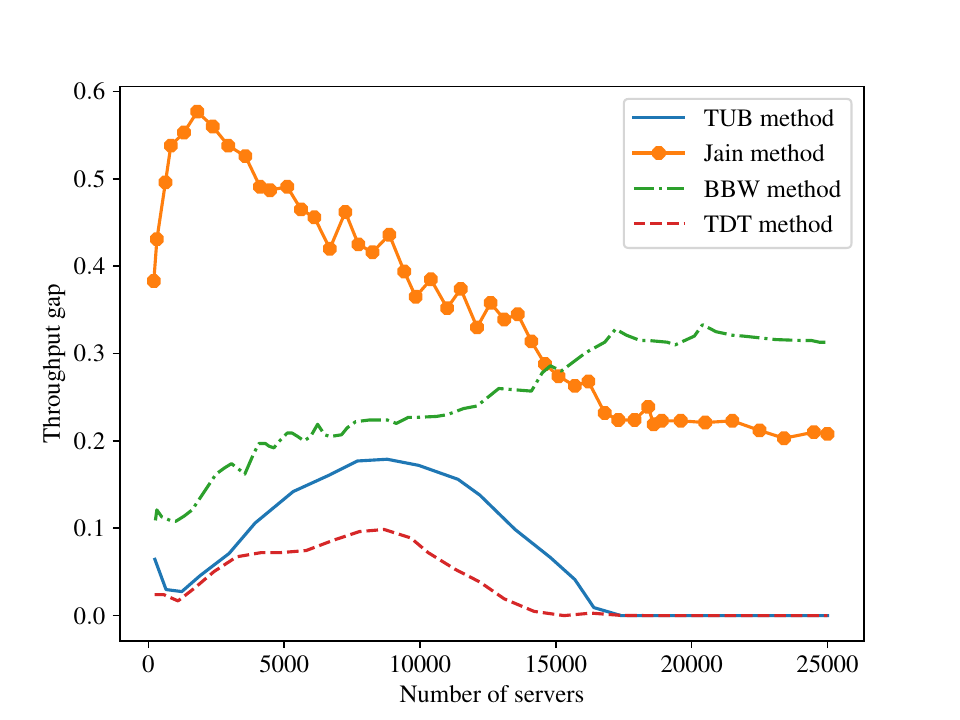}
			\caption{Jellyfish topology with $\mathscr{H} = 10$.}
			\label{fig:gap-j}
		\end{subfigure}%
		\begin{subfigure}{0.5\textwidth}
			\centering\includegraphics[scale=0.6]{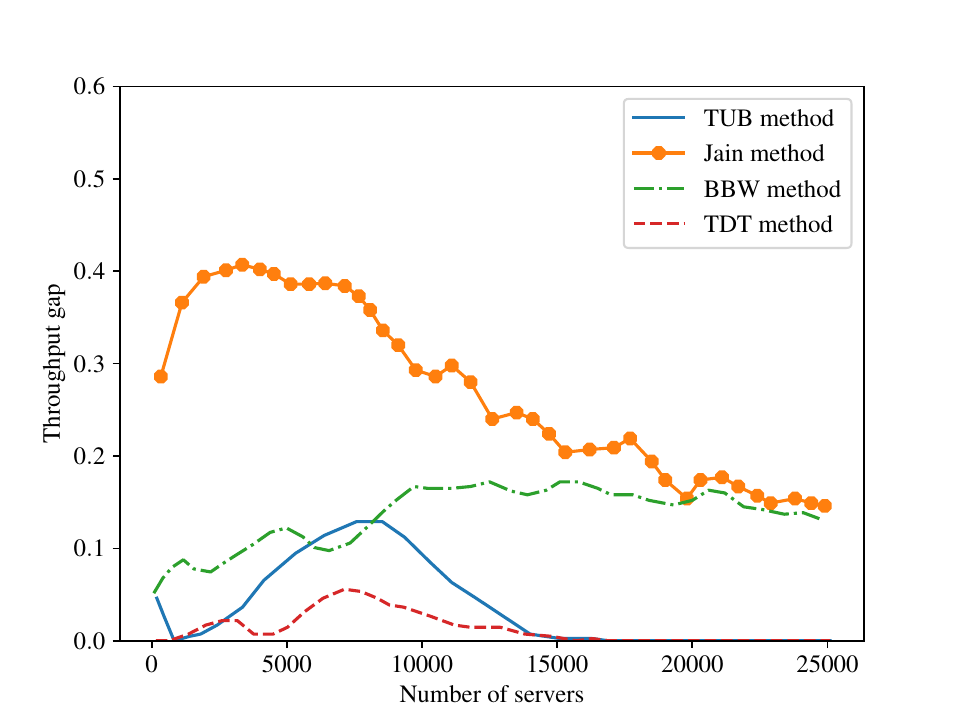}
			\caption{Fat-tree topology with $\mathscr{H}_{u} = 10$.}
			\label{fig:gap-f}
		\end{subfigure}
		\caption{Throughput gap reduction for datacenter networks using Traffic Divergence Theory.}
		\label{fig:manmade}
	\end{figure*}
	\begin{figure*}
		\centering
		\begin{subfigure}{0.5\textwidth}
			\centering\includegraphics[scale=0.6]{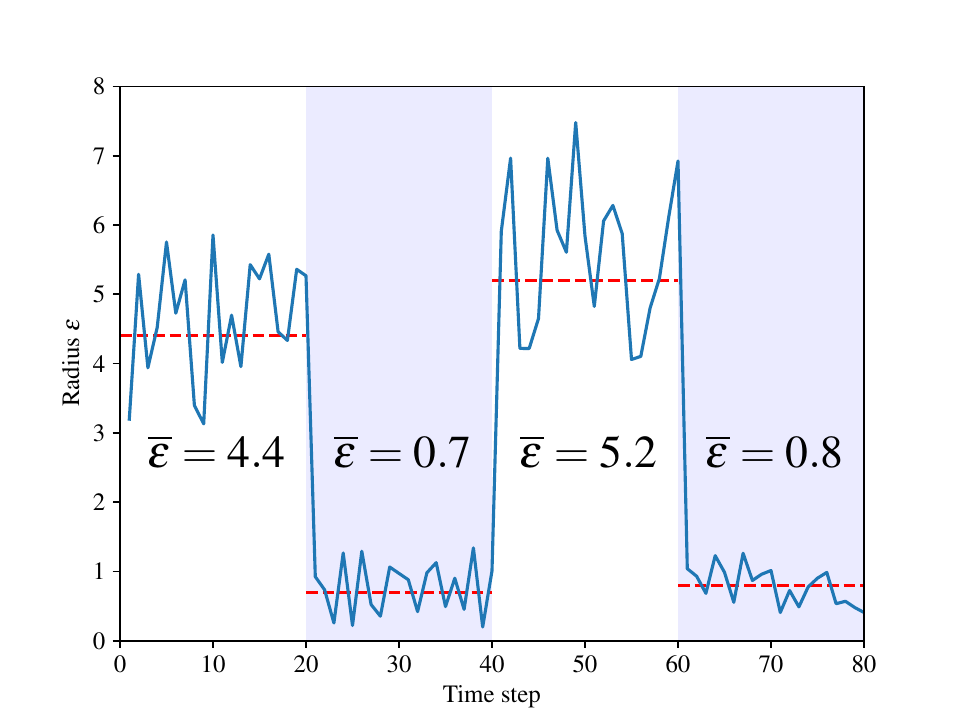}
			\caption{Jellyfish topology with $\mathscr{H} = 10$.}
			\label{fig:rad-j}
		\end{subfigure}%
		\begin{subfigure}{0.5\textwidth}
			\centering\includegraphics[scale=0.6]{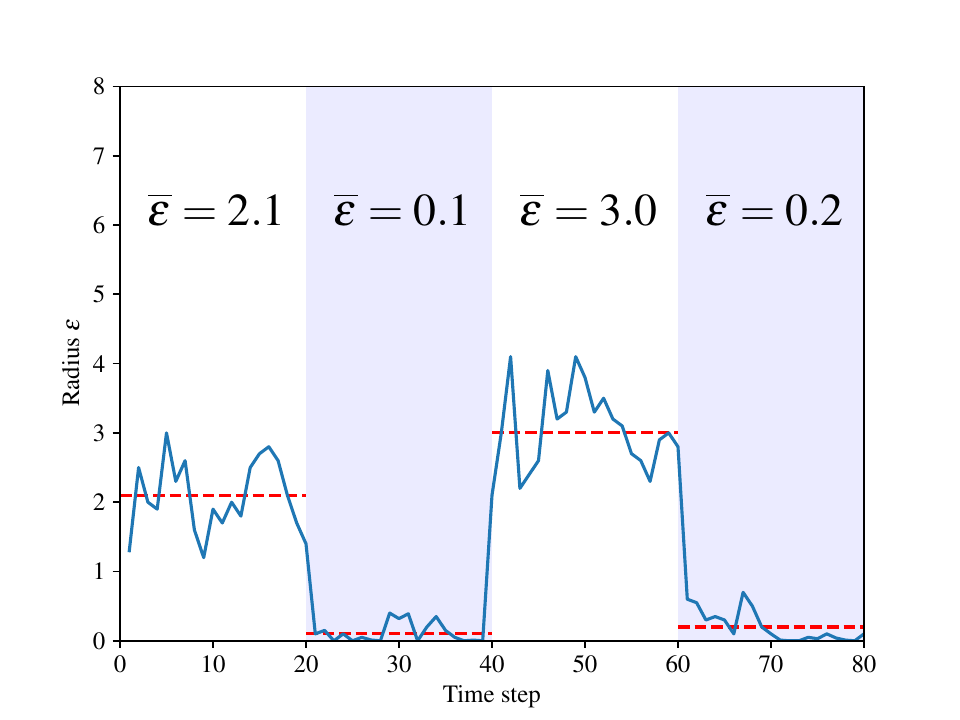}
			\caption{Fat-tree topology with $\mathscr{H}_{u} = 10$.}
			\label{fig:rad-f}
		\end{subfigure}
		\caption{Traffic distribution for datacenter networks using Traffic Divergence Theory. (The radius $\epsilon$ dynamics associated with a particular $\Delta_{u,v}$ is depicted. The average radius of each interval is noted on the figure.)}
		\label{fig:manmade1}
	\end{figure*}
	
	\z{We numerically validate that the result of Theorem \ref{thm:throughput} indeed reduces the throughput gap compared to Jain method \cite{jain2014maximizing}, BBW method \cite{kloti2015policy}, and TUB method \cite{namyar2021throughput}. As already noted, throughput gap is defined as the difference between the real throughput of a network topology and the predicted throughput based on the closed-form formulation in each of the quoted methods. We take Jellyfish topology \cite{singla2012jellyfish} and Fat-tree topology \cite{requena2008ruft} into account as typical uni-regular and bi-regular topologies, respectively. The total number of servers are selected according to the cardinality set $\{1, 500, 1000, \cdots, 25000\}$. Server-per-switch parameters are $\mathscr{H}\!=\!10$ (resp., $\mathscr{H}_{u}\!=\!8$) for uni-regular (resp., bi-regular) topology tests. Links are generated based on a uniform random distribution. In each simulation scenario, each entry $c_{u,v}$ of the capacity matrix $\mathscr{C}$ is 0 if there is no (in)direct link from $u$ to $v$, otherwise a positive number otherwise. The traffic matrix $\mathscr{T}$ of each scenario is randomly generated bounded by its corresponding capacity matrix. The profile of the coefficient set $\{\alpha_{\rho} \mid\forall \rho \in \psi(u,v) \}$ is based on the length of each contributing optimal route $\rho$, i.e., the longer the route is, the larger its coefficient is proportionally. This selection scheme roughly values the longer routes because of their larger capacity for traffic transfer.}
	
	\z{As Figure \ref{fig:gap-j} illustrates, the throughput gap associated with Jellyfish topology reduces, compared to the best available estimations in the literature, should one employs Algorithm 1 at the heart of which our traffic divergence theory functions. By increasing the number of servers, the gap curve asymptotically tends to zero because the number of optimal routes roughly increases. Thus, the contribution assignment of optimal routes, based on the coefficient set, become less relevant to the overall gap. A similar trend corresponding to Fat-tree topology may also be observed in Figure \ref{fig:gap-f}. The overall gap reduction is larger than that of Jellyfish topology since Fat-tree topology, similarly to other bi-regular topologies, includes more route options to fulfill maximal traffic distribution criteria of Algorithm 1.}
	
	\z{One can also validate the maximal distribution of traffic by the activation of Algorithm 1 leads to the traffic distribution maximization based on how radius $\epsilon$ varies in (\ref{eq:remm}). In particular, as depicted in Figure \ref{fig:rad-j}, we take a particular pair of nodes, labeled by $u$ and $v$, in the described Jellyfish topology to trace their relative traffic. In particular, given two typical nodes $u$ and $v$, the smaller the $\epsilon$ radius is, the more distributed their traffic are. We set up an experiment including 80 time steps partitioned to 4 equal intervals of 20 time steps. During the first and the third (resp., the second and the forth) intervals, we switch off (resp., on) the traffic distribution mechanism, i.e., Algorithm 1. So, we observe that the $\epsilon$ radii in the second and the fourth intervals are much closer to zero compared to those of the first and the third intervals. The same dynamics may be observed for any arbitrary pair of nodes. In a similar vein considering the same pair of nodes in the previous case, Figure \ref{fig:rad-f} also expresses the effectiveness of our formalism to distribute traffic around those nodes in Fat-tree topology. Following the discussion about the throughput gap, since Fat-tree topology generally includes more traffic-driven optimal routes, its radius $\epsilon$ is smaller than that of a Jellyfish topology.}
	\subsection{Power-Optimized Communication Planning for Robot Ad-hoc Networks}
	\label{subsec:adhoc}
	In ad-hoc robot networks, nodes are mobile robots whose movements are subject to the connectivity preservation of their links, which are point-to-point communicational channels between them. For such networks, the point-point assumption of the theory is also reasonable \cite{sinha2016throughput,saeed2021point,rahman2004controlled} given the current trend of beamforming with directional antennas \cite{ramanathan2001performance,dai2013overview}.
	
	Many mission-critical ad-hoc robot networks, including those used for firefighting \cite{kumar2004robot}, space operations \cite{clark2003dynamic}, etc., access very limited energy resources whose fast re-charge may be either impractical or very expensive. So, in the course of these missions, the least energy-consuming plan to perform a mission is the best one. In such ad-hoc robot networks, desired processing resources at a node directly depend on the traffic divergence of that node that quantifies its data transfer and processing load. So, given a desired throughput profile, and some power criteria, one seeks the minimum divergence for all nodes that is often equivalent to the least expensive hardware orchestration for that network. For this purpose, let $\mathscr{C}$ and $\mathscr{T}$ be the capacity and traffic matrices, respectively, associated with a typical ad-hoc robot network. Denote by $\theta^{\mathscr{T}}(\nabla_{u})$ the throughput profile of this network expressed in terms of the traffic divergence profiles of its nodes $(\forall u \in \mathscr{K})$. Moreover, assume a dynamics threshold $m$ to set a power limit for communication activities of this node. This limit determines an upper bound associated with the allowable power consumed by a node due to performance and energy considerations \cite{choi2017low}.
	
	Then, the quoted power-optimized communication problem can be formulated as an optimization problem as follows.
	\begin{figure}
		\centering\includegraphics[scale=1]{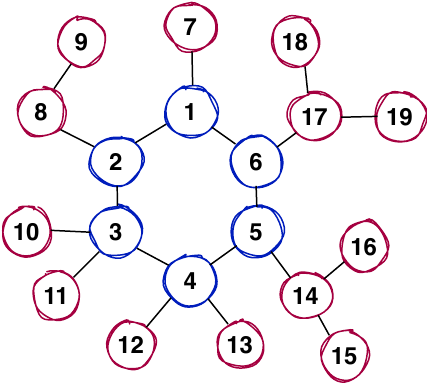}
		\caption{The ring topology associated with the hypothetical ad-hoc robot network.}
		\label{fig:topo}
	\end{figure}
	\begin{mini!}
		{\nabla_{u}}{\prod_{u\in\mathscr{K}}\nabla_{u}}{}{}\nonumber
		\addConstraint{\mathscr{T}\cdot\theta^{\mathscr{T}}(\nabla_{u})}{\le \mathscr{C}\quad}{(\forall u \in \mathscr{K})}\label{eq:cons1}
		\addConstraint{\Bigg|\frac{\partial \nabla_{u}}{\partial t}\Bigg|}{\le m}{(\forall u \in \mathscr{K})}\label{eq:cons2}
		\addConstraint{\nabla_{u}}{\in \mathbb{R}}{(\forall u \in \mathscr{K})}\nonumber
	\end{mini!}
	Here, (\ref{eq:cons1}) imposes any desired communication requirements on the planning process, while (\ref{eq:cons2}) governs communication activities of each node in terms of how much its traffic divergence rates are allowed to vary.
	We assume an ad-hoc robot network, including 19 nodes, with a ring topology, as shown in Figure \ref{fig:topo}. A ring topology includes a backbone ring and a set of branches \cite{macktoobian2023learning}. This topology simultaneously realizes optimal communication and fault tolerance for ad-hoc robot networks. Given a fix capacity matrix, we plan a traffic matrix in which each node has to communicate with all other nodes in that network. The doubly-stochastic entries of this matrix are weighted samples of a uniform distribution.
	
	Figure \ref{fig:opt} displays the solution profiles associated with two thresholds $m=1.1$ and $m=0.9$. Nodes \#1 to \#6 are the elements of the backbone cycle of the topology, as depicted in Figure \ref{fig:topo}, through which the majority of the network transient traffic is transmitted. Thus, it is expected that the optimal traffic divergences of these nodes are higher than those of the other nodes. Following the same line of reasoning, the least power burden associates with those nodes that have only one neighbor, e.g., node \#7. That is because such leaf nodes only process their own traffic but not any transient one. Accordingly, they may require less computational resources and energy consumption. Thus, their minimized traffic divergence values are relatively small. The performed simulations manifest that $m\le 0.8$ makes the mathematical program infeasible.
	\begin{figure}
		\centering\includegraphics[scale=0.7]{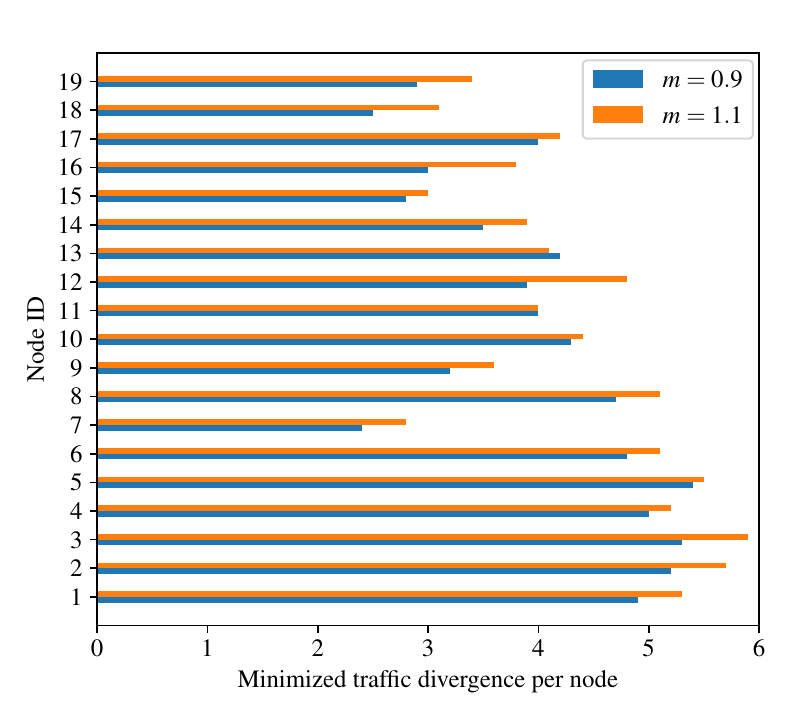}
		\caption{A typical solution to the power-optimized communication planning for an ad-hoc robot network}
		\label{fig:opt}
	\end{figure}
	\z{\section{Discussion}
		\label{sec:disc}
		We believe that this theory, given its general and efficiently-expressive formalism, may be capable of resolving even more complex problems in the scope of computer networks. Namely, it may provide new insights about less intuitive traffic-based mechanisms in complex networks, for example, similarly to what Theorem \ref{thm:nttr} asserts in view of entangled spatial and temporal traffic dynamics. With this aim, further research may be conducted to expand the theory by incorporating new ideas and analysis tools into its current formalism.}
	
	\z{The improvement of the solutions to the two example applications of the article may also be taken as other venues of research. In particular to the throughput gap problem, as already stated, the determination of coefficient set associated with the normalization of split ratios is a critical step for throughput predictions. Our strategy may be substituted with better metrics to set these coefficients for further reduction of the reported throughput gap. Regarding the power-optimized communication planning problem, one already notes the complexity of the nonlinear program that associates with the problem. The larger a target ad-hoc network is, the more problematic the complexity of this program will become in a computational point of view. To alleviate this issue, it would be beneficial to find equivalent, but less complex, constraints to simplify the program and its feasibility analysis for large-scale networks. Moreover, queuing features of nodes are abstracted away in this work, that is, the buffer size of queues are assumed to be infinite. However, one may incorporate more restrictive queue-dependent constraints into the current theory to realize a delay-dependent version of the current formalism for desired applications and analyses.}  Finally, this theory may provide new insights regarding control and behavior generation of multi-robot systems ranging from small teams \cite{lee2009decentralized} to large-scale swarms \cite{macktoobian2021astrobotics,macktoobian2019supervisory}. 
	\section{Concluding Remarks}
	\label{sec:conc}
	In this article, we presented a traffic divergence theory to model network dynamics in view of traffic variations. We elaborated on a computationally-rich set of properties associated with this theory that can express various spatial and temporal characterizations corresponding to traffic distribution. The theory may be generally applied to different types of networks for the analysis of those concepts that are entangled with traffic distribution. In this regard, we exhibited the wide applicability of the theory by its usage to model and solve two non-trivial problems: first, we improved the throughput gap corresponding to both uni-regular and bi-regular topologies; second, we planned power-optimized communication for ad-hoc robot networks.
	
	\section*{Appendix} 
	\appendix  
	\section{The Proof of Proposition \ref{prop:ltc}} 
	\label{app:ltc}
	As sink flows toward the nodes positively contribute to the link TD except those passed between them, the net sink flow may be written as
	\begin{equation}
		\nabla_{u,v}^{\leftarrow} = \sum_{z \in \mathscr{N}_{u}\setminus\{v\}}^{} [f_{zu} - f_{vu} ] + \sum_{z \in \mathscr{N}_{v}\setminus\{u\}}^{} [f_{zv} - f_{uv}].
	\end{equation}
	Similarly to the sink case, source flows leaving the nodes negatively contribute to the net divergence of their link while the direct flow between the nodes do not participate in the net term because they are not interface flows of the link but internal ones between the nodes: 
	\begin{equation}
		\nabla_{u,v}^{\rightarrow} = \sum_{z \in \mathscr{N}_{u}\setminus\{v\}}^{} [f_{uz} - f_{uv} \Big] + \sum_{z \in \mathscr{N}_{v}\setminus\{u\}}^{} [f_{vz} - f_{vu}].
	\end{equation}
	Thus, one yields 
	\begin{equation}
		\begin{split}
			\nabla_{u,v} =& \nabla_{u,v}^{\leftarrow} - \nabla_{u,v}^{\rightarrow}\\=& \sum_{z \in \mathscr{N}_{u}}^{} [f_{zu} - f_{uz}] + \sum_{z \in \mathscr{N}_{v}}^{} [f_{zv} - f_{vz}]
			\\=&  \nabla_{u} + \nabla_{v},
		\end{split}
	\end{equation}
	which completes the proof.
	\section{The Proof of Proposition \ref{prop:rrtd}}
	\label{app:rrtd}
	If the length of $\Omega$ is 2, i.e., $\Omega$ being a link, then its TD is trivially the link TD, according to Proposition \ref{prop:ltc}. As the first non-trivial case, assume that $\Omega$ is constituted by nodes $u,v,w \in \mathscr{V}$ such that $\Omega$ is a route between $u$ and $v$ where $w$ is the exclusive node between them. Then, writing the equations associated with the sink and source components of $\nabla_{u,w}$ and $\nabla_{w,v}$ then following the same line as used in the proof of Proposition \ref{prop:ltc} yields $
	\nabla_{\Omega} = \nabla_{u} + \nabla_{w} + \nabla_{v}$. The claim of this proposition can be obtained by mathematical induction on the length of $\Omega$.
	\section{The Proof of Proposition \ref{prop:nlrtd}}
	\label{app:nlrtd}
	Using Propositions \ref{prop:ltc} and \ref{prop:rrtd} yields
	\begin{equation}
		\begin{split}
			\nabla_{\Omega} &= \nabla_{u} + \nabla_{w} + \nabla_{v} \\&= (\nabla_{u} + \nabla_{w}) + \nabla_{v} = \nabla_{u,w} + \nabla_{v} \\&= (\nabla_{w} + \nabla_{v}) + \nabla_{u} = \nabla_{w,v} + \nabla_{u}.
		\end{split}
	\end{equation}
	\section{The Proof of Theorem \ref{thm:dtd}}
	\label{app:dtd}
	We prove the claim by induction of $n$. For the base case associated with $n=1$, we start from the spatial link-node TD derivative.
	\begin{align}
		\label{eq:ttt1}
		\frac{\partial\nabla_{u,v}}{\partial\nabla_{u}} = \frac{\partial(\nabla_{u} + \nabla_{v})}{\partial\nabla_{u}} = 1 + \frac{\partial \nabla_{v}}{\partial\nabla_{u}}\\
		\label{eq:ttt2}
		\frac{\partial\nabla_{u,v}}{\partial\nabla_{v}} = \frac{\partial(\nabla_{u} + \nabla_{v})}{\partial\nabla_{v}} = 1 + \frac{\partial \nabla_{u}}{\partial\nabla_{v}}
	\end{align}
	Subtraction of the second equation from the first one yields
	\begin{equation}
		\frac{\partial\nabla_{u,v}}{\partial\nabla_{u}} - \frac{\partial\nabla_{u,v}}{\partial\nabla_{v}} = [\frac{\partial}{\partial \nabla_{u}} - \frac{\partial}{\partial \nabla_{v}}]\nabla_{u,v} = \frac{\partial \nabla_{v}}{\partial \nabla_{u}} - \frac{\partial \nabla_{u}}{\partial \nabla_{v}}.
	\end{equation}
	Arbitrary number of successive partial derivations applied to (\ref{eq:ttt1}) and (\ref{eq:ttt2}) leads to what follows.
	\begin{align}
		\frac{\partial^{n}\nabla_{u,v}}{{\partial\nabla_{u}}^{n}} =
		\frac{\partial^{n}(\nabla_{u} + \nabla_{v})}{{\partial\nabla_{u}}^{n}} = 1+	\frac{\partial^{n} \nabla_{v}}{{\partial\nabla_{u}}^{n}}\\
		\frac{\partial^{n}\nabla_{u,v}}{{\partial\nabla_{v}}^{n}} = \frac{\partial^{n}(\nabla_{u} + \nabla_{v})}{{\partial\nabla_{v}}^{n}}=1+ \frac{\partial^{n} \nabla_{u}}{{\partial\nabla_{v}}^{n}}
	\end{align}
	Similarly to the base case, subtracting the second equation from the first one yields the expected result.
	\section{The Proof of Theorem \ref{thm:nttr}}
	\label{app:nttr}
	Since
	\begin{equation}
		\label{eq:main}
		\frac{\partial \nabla_{u}}{\partial t} = \sum_{z \in \mathscr{N}_{u}} \frac{\partial \nabla_{u}}{\partial \nabla_{z}}\frac{\partial \nabla_{z}}{\partial t}
	\end{equation}
	Using Cauchy-Schwartz inequality, one yields
	\begin{equation}
		\Big[\sum_{z \in \mathscr{N}_{u}} \frac{\partial \nabla_{u}}{\partial \nabla_{z}}\frac{\partial \nabla_{z}}{\partial t}\Big]^{2}\le \sum_{z \in \mathscr{N}_{u}} \Big[\frac{\partial \nabla_{u}}{\partial \nabla_{z}}\Big]^{2}\sum_{z \in \mathscr{N}_{u}} \Big[\frac{\partial \nabla_{z}}{\partial t}\Big]^{2}.
	\end{equation}
	Given the following elementary result, where $\{x_{i}\}$ is a finite series of integer numbers including $m$ elements whose average is $\overline{x}$
	\begin{equation}
		\label{eq:stat}
		\sum_{i=1}^{m}x_{i} = m(\frac{1}{m}\sum_{i=1}^{m}x_{i})^{2} + \frac{1}{m}\sum_{i=1}^{m} (x - \overline{x})^{2},
	\end{equation}
	we define average spatial and temporal traffic rates at $u$ as follows.
	\begin{align}
		\partial_{u,z} &\defeq n^{-1}\sum\limits_{z \in \mathscr{N}_{u}}\frac{\partial \nabla_{u}}{\partial \nabla_{z}}\label{eq:ave_1}\\
		\partial_{z,t} &\defeq n^{-1}\sum\limits_{z \in \mathscr{N}_{u}}\frac{\partial \nabla_{z}}{\partial t}\label{eq:ave_2}
	\end{align}
	By writing (\ref{eq:stat}) with respect to $\frac{\partial \nabla_{u}}{\partial \nabla_{z}}$ using (\ref{eq:ave_1}), the first term in the RHS of (\ref{eq:main}) may be written as follows.
	\begin{align}
		\sum_{z \in \mathscr{N}_{u}}\Bigl[\frac{\partial \nabla_{u}}{\partial \nabla_{z}}\Bigr]^{2}&=  \frac{1}{n}\Bigl[\Bigl(\sum_{z \in \mathscr{N}_{u}}\frac{\partial \nabla_{u}}{\partial \nabla_{z}}\Bigr)^{2}+\sum_{z \in \mathscr{N}_{u}}\Bigl(\frac{\partial \nabla_{u}}{\partial \nabla_{z}}-\partial_{u,z}\Bigr)^{2}\Bigr]\nonumber\\
		&= \frac{1}{n}\Bigl[\Bigl(\sum_{z \in \mathscr{N}_{u}}\frac{\partial \nabla_{u}}{\partial \nabla_{z}}\Bigr)^{2}+\sum_{z \in \mathscr{N}_{u}}\Bigl(\frac{\partial \nabla_{u}}{\partial \nabla_{z}}-\frac{1}{n}\sum_{z' \in \mathscr{N}_{u}}\frac{\partial \nabla_{u}}{\partial \nabla_{z'}}\Bigr)^{2}\Bigr]
		\nonumber\\&=
		\frac{1}{n}\Bigl[\Bigl(\sum_{z \in \mathscr{N}_{u}}\frac{\partial \nabla_{u}}{\partial \nabla_{z}}\Bigr)^{2}+\sum_{z \in \mathscr{N}_{u}}
		\Bigl[
		\Bigl(\frac{\partial \nabla_{u}}{\partial \nabla_{z}}\Bigr)^{2}+
		\frac{1}{n^{2}}\Bigl(\sum_{z' \in \mathscr{N}_{u}}\frac{\partial \nabla_{u}}{\partial \nabla_{z'}}\Bigr)^{2}-
		\frac{2}{n}\Bigl(\frac{\partial \nabla_{u}}{\partial \nabla_{z}}\Bigr)\Bigl(\sum_{z' \in \mathscr{N}_{u}}\frac{\partial \nabla_{u}}{\partial \nabla_{z'}}\Bigr)\Bigr]
		\Bigr]\nonumber\\&= \frac{1}{n}\Bigl[\Bigl(\sum_{z \in \mathscr{N}_{u}}\frac{\partial \nabla_{u}}{\partial \nabla_{z}}\Bigr)^{2}
		+\sum_{z \in \mathscr{N}_{u}}
		\Bigl(\frac{\partial \nabla_{u}}{\partial \nabla_{z}}\Bigr)^{2}+
		\frac{1}{n^{2}}\sum_{z \in \mathscr{N}_{u}}\Bigl(\sum_{z' \in \mathscr{N}_{u}}\frac{\partial \nabla_{u}}{\partial \nabla_{z'}}\Bigr)^{2}
		-
		\frac{2}{n}\sum_{z \in \mathscr{N}_{u}}\Bigl(\frac{\partial \nabla_{u}}{\partial \nabla_{z}}\Bigr)\Bigl(\sum_{z' \in \mathscr{N}_{u}}\frac{\partial \nabla_{u}}{\partial \nabla_{z'}}\Bigr)
		\Bigr]\label{eq:last}
	\end{align}
	In the RHS of the equation above, the outer summation in the third term has a free index. So, that summation is equivalent to the multiplication of the cardinality of its index set, i.e., $n$, to its subsequent summation. Then, index $z'$ of the inner summation may be renamed to $z$. In the forth term, we have two identical summations multiplied to each other. So, (\ref{eq:last}) may be further simplified as below.
	\begin{equation*}
		\begin{split}
			\sum_{z \in \mathscr{N}_{u}} \Bigl[\frac{\partial \nabla_{u}}{\partial \nabla_{z}}\Bigr]^{2}&= \frac{1}{n}\Bigl[\Bigl(\sum_{z \in \mathscr{N}_{u}}\frac{\partial \nabla_{u}}{\partial \nabla_{z}}\Bigr)^{2}
			+
			\sum_{z \in \mathscr{N}_{u}}
			\Bigl(\frac{\partial \nabla_{u}}{\partial \nabla_{z}}\Bigr)^{2}
			+
			\frac{1}{n}\Bigl(\sum_{z \in \mathscr{N}_{u}}\frac{\partial \nabla_{u}}{\partial \nabla_{z}}\Bigr)^{2}
			-
			\frac{2}{n}\Bigl(\sum_{z \in \mathscr{N}_{u}}\frac{\partial \nabla_{u}}{\partial \nabla_{z}}\Bigr)^{2}
			\Bigr]\\
			&=\frac{1}{n}\Bigl[\Bigl(\sum_{z \in \mathscr{N}_{u}}\frac{\partial \nabla_{u}}{\partial \nabla_{z}}\Bigr)^{2}	
			\sum_{z \in \mathscr{N}_{u}}
			\Bigl(\frac{\partial \nabla_{u}}{\partial \nabla_{z}}\Bigr)^{2}
			-
			\frac{1}{n}\Bigl(\sum_{z \in \mathscr{N}_{u}}\frac{\partial \nabla_{u}}{\partial \nabla_{z}}\Bigr)^{2}
			\Bigr]
		\end{split}
	\end{equation*}
	One can take the definition of spatial traffic divergence rate into account so that the above equation turns to 
	\begin{equation}
		\sum_{z \in \mathscr{N}_{u}} \bigl[\frac{\partial \nabla_{u}}{\partial \nabla_{z}}\bigr]^{2} = n^{-1}\square_{u}^{2}.
	\end{equation} 
	A similar argument for the second term in the RHS of (\ref{eq:main}) yields 
	\begin{equation}
		\sum_{z \in \mathscr{N}_{u}} \bigl[\frac{\partial \nabla_{z}}{\partial t}\bigr]^{2} = n^{-1}\boxplus_{u}^{2}.
	\end{equation}. 
	So, we have
	\begin{equation}
		\bigl[\frac{\partial \nabla_{u}}{\partial t}\bigr]^{2} \le n^{-2}\square_{u}^{2}\boxplus_{u}^{2},
	\end{equation} 
	that leads to the desired result, i.e., 
	\begin{equation}
		-n^{-1}[\square_{u}\boxplus_{u}]\le \frac{\partial \nabla_{u}}{\partial t} \le n^{-1}[\square_{u}\boxplus_{u}].
	\end{equation}
	\section{The Proof of Proposition \ref{prop:mtdc}}
	\label{app:mtdc}
	According to Theorem \ref{thm:dtd}, differential traffic divergence dynamics (\ref{eq:xxxx}), in particular the case $n=1$, may be rearranged as below for $u$ and $v$ and their adjacent nodes in $\mathscr{N}_u$ and $\mathscr{N}_v$.
	\begin{align}
		\frac{\partial\nabla_{u}}{\partial\nabla_{z}} = \frac{\partial\nabla_{z}}{\partial\nabla_{u}} + [\frac{\partial}{\partial\nabla_{u}} - \frac{\partial}{\partial\nabla_{z}}]\nabla_{u,z}\quad(\forall z \in \mathscr{N}_{u})\\
		\frac{\partial\nabla_{v}}{\partial\nabla_{z}} = \frac{\partial\nabla_{z}}{\partial\nabla_{v}} + [\frac{\partial}{\partial\nabla_{v}} - \frac{\partial}{\partial\nabla_{z}}]\nabla_{v,z}\quad(\forall z \in \mathscr{N}_{v})
	\end{align}
	Substituting these two equations into (\ref{eq:mtdc}) completes the proof.
	\section{The Proof of Lemma \ref{lem:symm}}
	\label{app:symm}
	An immediate consequence of $\Delta_{u,v} = \Delta_{w,v} = \gamma$ is $\Delta_{u,w} = \Delta_{w,u} = 1$, which is equivalent to the maximal traffic distribution between $u$ and $v$.
	\section{The Proof of Theorem \ref{thm:loc}}
	\label{app:loc}
	The case $m=1$ in Lemma \ref{lem:symm} corresponds to traffic symmetry under maximal traffic distribution. So, for any arbitrary number of nodes, if the relative traffic rate between every two neighboring node is maximally distributed, then any relative traffic rate between two arbitrary distant nodes has to also be maximally distributed.
	\section{The Proof of Proposition \ref{prop:cc}}
	\label{app:cc}
	In the worst case, the length of an optimal route is $n$. Thus, the computational complexity of the do-while block of the algorithm is $\mathcal{O}\big(n^{3}\big)$. The while loop may be executed as many as $n-1$ times. In particular, Each of Lines 9, 11, 12, 13 may be executed in one operation. Defining the maximum number of neighbors associated with a node as 
	\begin{equation}
		Y \defeq \max_{z \in \mathscr{K}} \mathscr{N}_{z},
	\end{equation} 
	Line 10 requires $Y$ operations. Thus, since in general $n\gg Y$, the overall computational complexity reads as $	\mathcal{O}\big(n^{3}+(n-1)(Y+4)\big)\equiv\mathcal{O}\big(n^{3})$, which is what the proposition claims.
	\section{The Proof of Theorem \ref{thm:throughput}}
	\label{app:throughput}
	(i) \textit{Uni-regular case}: The traffic reaching a switch $u \in \mathscr{K}$ belongs to either of the following categories: it is destined to reach $u$ as its target, or it corresponds to transient traffic. The cited transient traffic $\Lambda_{u}^\mathscr{T}$ is constrained with respect to the capacity of the switch that, given the total number of switches $\mathscr{R}_{u}$ and the number of outage ports $\mathscr{Q}_{u}$, reads as follows.
	\begin{equation}
		\label{eq:temp3}
		\Lambda_{u}^\mathscr{T} + \theta^\mathscr{T}\sum_{v \in \mathscr{K}\setminus\{u\}} \mathscr{T}_{uv} = \mathscr{R}_{u} - \mathscr{Q}_{u} - \mathscr{H}
	\end{equation}  
	Consequently, noting that \cite{namyar2021throughput}
	\begin{equation}
		\label{eq:temp4}
		\sum_{u \in \mathscr{K}}[\mathscr{R}_{u} - \mathscr{Q}_{u} - \mathscr{H}] = 2\mathscr{E},
	\end{equation}
	the total transient traffic of the topology is obtained as below.
	\begin{equation}
		\label{eq:first}
		\sum_{u \in \mathscr{K}} \Lambda_{u}^{\mathscr{T}} = 2\mathscr{E} - \theta^{\mathscr{T}}\sum_{u \in \mathscr{K}}\sum_{v \in \mathscr{K}\setminus\{u\}}\mathscr{T}_{uv}
	\end{equation}
	Alternatively, one may compute the total transient traffic of the topology based the route set $\psi_{u,v}$ between arbitrary switches $u,v \in \mathscr{K}$ computed by Algorithm 1, i.e., the congestion-avoiding minimum-hop routing mechanism.
	\begin{equation}
		\label{eq:x}
		\sum_{u \in \mathscr{K}}\Lambda_{u}^{\mathscr{T}} = \theta^{\mathscr{T}}\sum_{u \in \mathscr{K}}\sum_{v \in \mathscr{K}\setminus\{u\}} t_{uv} \sum_{\rho \in \psi_{u,v}} \beta_{\rho}^{\mathscr{T}}(|\rho| - 1)
	\end{equation} 
	Here, $|\rho|$ is the length of route $\rho$ and $\beta_{\rho}^{\mathscr{T}}$ is the split ratio of $\rho$. This ratio determines the contribution of this route to the carriage of the cited total transient traffic. The lower the traffic of a switch is, the higher its split ratio shall be to minimize the likelihood of any congestion, implying 
	\begin{equation}
		\beta_{\rho}^{\mathscr{T}} \propto \nabla_{\rho}^{-1},
	\end{equation}
	that using the result of Proposition \ref{prop:rrtd} leads to
	\begin{equation}
		\beta_{\rho}^{\mathscr{T}} \propto \big(\sum_{w \in \rho }\nabla_{w}\big)^{-1}.
	\end{equation} 
	The overall contributions of split ratios associated with all found routes shall be equal to unit. Thus, the co-efficient set (\ref{eq:coeff}), one may write 
	\begin{equation}
		\sum_{\rho \in \psi_{u,v}}\beta_{\rho}^{\mathscr{T}} = \alpha_{\rho}\big(\sum_{w \in \rho }\nabla_{w}\big)^{-1} = 1,
	\end{equation}
	which turns (\ref{eq:x}) to
	\begin{equation}
		\label{eq:second}
		\begin{split}
			\sum_{u \in \mathscr{K}}\Lambda_{u}^{\mathscr{T}} = \theta^{\mathscr{T}}\sum_{u \in \mathscr{K}}\sum_{v \in \mathscr{K}\setminus\{u\}} \mathscr{T}_{uv} \Big[\smashoperator[r]{\sum_{\rho \in \psi_{u,v}}} &\alpha_{\rho}\Big(\sum_{w \in \rho}\nabla_{w}\Big)^{-1}\\&\Big(|\rho| - 1\Big)\Big].
		\end{split}
	\end{equation}
	The double stochasticity of $\mathscr{T}$ implies that the summations of all of its rows and columns are equal to $\mathscr{H}$. Finally, by the equality of (\ref{eq:first}) and (\ref{eq:second}) and some simplifications, we obtain
	\begin{equation}
		\theta^\mathscr{T} = \frac{2\mathscr{E}}{\mathscr{H}}\Biggl[1 + \smashoperator[l]{\sum\limits_{\substack{u \in \mathscr{K}, v \in \mathscr{K}\setminus\{u\},\\\rho \in \psi_{u,v}}}} \Bigg[\alpha_{\rho}|\rho| \Big[\sum_{w \in \rho}^{} \nabla_{w}\Big]^{-1}\Bigg] - |\psi_{u,v}| \Biggr]^{-1},
	\end{equation}
	which is the theorem's claim.
	
	\noindent(ii) \textit{Bi-regular case}: The capacity of each switch is variable depending on whether it is connected to servers or solely to other switches. Formally speaking, given the selector function below applied to an arbitrary switch $u \in \mathscr{K}$  
	\begin{equation}
		\mathscr{S}(u) \defeq
		\begin{cases}
			1\quad \text{if } u \text{ is connected to any servers}\\
			0\quad \text{otherwise}
		\end{cases},
	\end{equation}  
	(\ref{eq:temp3}) is written as 
	\begin{equation}
		\Lambda_{u}^\mathscr{T} + \theta^\mathscr{T}\sum_{v \in \mathscr{K}\setminus\{u\}} \mathscr{T}_{uv} = \mathscr{R}_{u} - \mathscr{Q}_{u} - \mathscr{H}_{u}\mathscr{S}(u).
	\end{equation}
	Then, (\ref{eq:temp4}) turns into 
	\begin{equation}
		\sum_{u \in \mathscr{K}}[\mathscr{R}_{u} - \mathscr{Q}_{u} - \mathscr{H}_{u}\mathscr{S}(u)] = 2\mathscr{E}.
	\end{equation}
	However, (\ref{eq:first}) is invariant with respect to this change. So, the remainder of the proof for the bi-regular case follows that of the uni-regular one.
\bibliographystyle{IEEEtran} 
\bibliography{references}
\end{document}